\documentclass[envcountsect,envcountsame]{llncs}
\usepackage{amsfonts}
\usepackage{amssymb,epic,eepic}
\usepackage{amsmath}
\usepackage{enumerate}
\usepackage[usenames,dvipsnames]{pstricks}
\usepackage{epsfig}
\usepackage{pst-grad} % For gradients
\usepackage{pst-plot} % 

\newcommand{\ty}{{\mathcal{T}}}
\newcommand{\pp}{{\mathcal{P}}}

\newcommand{\nn}{{\mathbb N}}

\begin{document}

\title{Capacity Results for Arbitrarily Varying Wiretap Channels}
\author{Igor Bjelakovi\'c, Holger Boche and Jochen Sommerfeld}
\institute{Lehrstuhl f\"ur theoretische Informationstechnik, Technische Universit\"at M\"unchen, 80290
  M\"unchen, Germany\\ \email{\{igor.bjelakovic, boche, jochen.sommerfeld\}@tum.de}}
%\dedicatory{Dedicated to the memory of Rudolf Ahlswede} 
\maketitle
\begin{center}
\emph{Dedicated to the memory of Rudolf Ahlswede}
\end{center}

\begin{abstract}
In this work the arbitrarily varying wiretap channel AVWC is studied. We derive a lower bound on the
random code secrecy capacity for the average error criterion and the strong secrecy criterion in the case of a
best channel to the eavesdropper by using Ahlswede's robustification technique for ordinary AVCs. We show
that in the case of a non-symmetrisable channel to the legitimate receiver the deterministic code secrecy
capacity equals the random code secrecy capacity, a result similar to Ahlswede's dichotomy result for
ordinary AVCs. Using this we can derive that the lower bound is also valid for the deterministic code
capacity of the AVWC. The proof of the dichotomy result is based on the elimination technique introduced
by Ahlswede for ordinary AVCs.  We further prove upper bounds on the deterministic code secrecy capacity
in the general case, which results in a multi-letter expression for the secrecy capacity in the case of a
best channel to the eavesdropper. Using techniques of Ahlswede, developed to guarantee the validity of a
reliability criterion, the main contribution of this work is to integrate the strong secrecy criterion
into these techniques.    
\end{abstract}

\section{Introduction}
Models of communication systems taking into account both the requirement of security against a potential eavesdropper
and reliable information transmission to legitimate receivers which suffer from channel uncertainty, have
received much interest in current research. One of the simplest communication models with channel uncertainty are
compound channels, where the channel realisation remains fixed during the whole transmission of a
codeword. Compound wiretap channels were the topic of previous work of the authors \cite{bjela3},
\cite{bjela2} and for example of \cite{liang}, \cite{bloch}. In the model of an arbitrarily varying wiretap
channel AVWC the channel state to both the legitimate receiver and the eavesdropper varies from symbol to
symbol in an unknown and arbitrary manner. Thus apart from eavesdropping the model takes into account an active
adversarial jamming situation in which the jammer chooses the states at her/his will. Then reliable transmission to the
legitimate receiver must be guaranteed in the presence of the jammer. 

In this paper we consider families of pairs of channels $\mathfrak{W}=\{(W_{s^n},V_{s^n}):s^n \in S^n\}$
with common input alphabets and possibly different output alphabets, where $s^n \in S^n$ denotes the state
sequence during the transmission of a codeword. The legitimate users are connected via $W_{s^n}$
and the eavesdropper observes the output of $V_{s^n}$. In our communication scenario the legitimate users
have no channel state information. We derive capacity results for the AVWC $\mathfrak{W}$ under the
average error probability criterion and a strong secrecy criterion. The investigation of the
corresponding problem concerning the maximum error criterion is left as a subject of future
investigations. However, we should emphasize that there is no full capacity result for an ordinary
(i.e. without secrecy constraints) AVC for the maximum error criterion. Together with Wolfowitz in
\cite{ahlswwolf2} Ahlswede determined the capacity for AVCs with binary output alphabets under this criterion. In
\cite{ahlsw4} he showed that the general solution is connected to Shannon's zero error capacity problem
\cite{shannon2}. 

Two fundamental techniques, called \emph{elimination} and \emph{robustification technique} discovered by
Ahlswede will play a crucial role in this paper. In \cite{ahlsw3} he developed the
\emph{elimination technique} to derive the deterministic code capacity for AVCs under the average
error probability criterion, which is either zero or equals its random code capacity, a result, which is
called Ahlswede's dichotomy for single user AVCs. With the so-called
\emph{robustification technique} \cite{ahlsw2} in turn he could link random codes for the AVC to deterministic
codes for compound channels.  
Further in the papers \cite{ahlswcsis1}, \cite{ahlswcsis2} on \emph{common randomness} in information
theory Ahlswede together with Csiszar studied, inter alia, problems of information theoretic security by
considering a model which enables secret sharing of a random key, in particular in the presence of a
wiretapper.  Because the arbitrarily varying wiretap channel AVWC combines both the wiretap channel and
the AVC it is not surprising that we can use the aforementioned techniques to derive capacity results for
the AVWC. The actual challenge of our work was to integrate the strong secrecy criterion in both the
\emph{elimination} and the \emph{robustification technique}, approaches, both were developed to guarantee
a reliability criterion. As it was shown in \cite{bjela2}, compared with weaker secrecy criteria, the
strong secrecy criterion ensures that the average error probability of every decoding strategy  of the
eavesdropper in the limit tends to one. 

In Section \ref{Rancode} we give a lower bound on the random code secrecy capacity in the special
case of a "best" channel to the eavesdropper. The proof is based on the \textit{robustification
  technique} by Ahlswede \cite{ahlsw2} combined with results for compound wiretap
channels given by the authors in \cite{bjela2}. 

In Section \ref{determ_code} we
use the \textit{elimination technique} \cite{ahlsw3}, which is composed of the \textit{random code
  reduction}  and the \textit{elimination of randomness} \cite{csis2}, to show that, provided that the
channel to the legitimate receiver is non-symmetrisable, the deterministic code secrecy capacity equals
the random code secrecy capacity and to give a condition when it is greater than zero. Thus we establish
a result for the AVWC that is similar to that of Ahlswede's dichotomy result for ordinary AVCs. As a consequence the
above-mentioned lower bound on the random code secrecy capacity can be achieved by a deterministic code
under the same assumptions.

In Section \ref{upper_bound} we give a single-letter upper bound on the deterministic code secrecy
capacity, which corresponds to the upper bound of the secrecy capacity of a compound wiretap
channel. Moreover, by establishing an multi-letter upper bound on the secrecy capacity we can conclude to
a multi-letter expression of the secrecy capacity of the AVWC in the special case of a best channel to the
eavesdropper.
 
The lower bound on the secrecy capacity as well as other results were given earlier
in \cite{molav} for a weaker secrecy criterion, but the proof techniques for the stronger
secrecy criterion differ significantly, especially in the achievability part for the random
codes.
 
%----------------------------------------------------------------------------------------------------------------------------------------------------------
%--------------------------------------------------------Model-----------------------------------------------------------------------------------------
%----------------------------------------------------------------------------------------------------------------------------------------------------------
\section{Arbitrarily Varying Wiretap Channels}\label{avwc}
\subsection{Definitions}

Let $A,B,C$ be finite sets and consider a non-necessarily finite family of channels 
$W_s:A\to\mathcal{P}(B)$\footnote{$\mathcal{P}(B)$ denotes the set of probability distributions on $B$.},
where $s \in S$ denotes the state of the channel. Now, given $s^n=(s_1,s_2, \ldots ,s_n) \in S^n$ we
define the stochastic matrix
\begin{equation}\label{eq:1}
W^n(y^n| x^n, s^n):=\prod^n_{i=1} W(y_i| x_i, s_i) := \prod^n_{i=1} W_{s_i}(y_i| x_i)
\end{equation}
for all $y^n=(y_1, \ldots, y_n) \in B^n$ and $x^n=(x_1, \ldots ,x_n) \in A^n$. An arbitrarily varying
channel is then defined as the sequence $\{\mathcal{W}^n \}^\infty_{n=1}$ of the family of channels
$\mathcal{W}^n=\{W^n(\cdot| \cdot, s^n) : s^n \in S^n\}$. Now let $\mathcal{W}^n$ represent the
communication link to a legitimate receiver to which the transmitter wants to send a private message,
such that a possible second receiver should be kept as ignorant of that message as possible. We call this
receiver the eavesdropper, which observes the output of a second family of channels
$\mathcal{V}^n=\{V^n(\cdot| \cdot, s^n) : s^n \in S^n\}$ with an analogue definition of $V^n(\cdot|
\cdot, s^n)$ as in \eqref{eq:1} for $V_s:A\to \mathcal{P}(C)$, $s \in S$. Then we denote the set of the two
families of channels with common input by $\mathfrak{W}=\{(W_{s^n},V_{s^n}):s^n \in S^n\}$ and call it
the arbitrarily varying wiretap channel. In addition, we assume that the state sequence $s^n$ is unknown
to the legitimate receiver, whereas the eavesdropper always knows which channel is in use. 

%--------------------------------------stochastic_encoder----------------------------------------------------------------------------------------
%---------------------------------------------------------------------------------------------------------------------------------------------------------
A $(n,J_n)$ code $\mathcal{C}_n$ for the arbitrarily varying wiretap channel $\mathfrak{W}$ consists of a stochastic encoder 
$E:\mathcal{J}_n\to \mathcal{P}(A^n)$ (a stochastic matrix) with a message set
$\mathcal{J}_n:=\{1,\ldots, J_n \}$ and a collection of mutually 
disjoint decoding sets $\{D_j\subset B^n:j\in\mathcal{J}_n  \}$. The average error probability of a
code $\mathcal{C}_n$ is given by
\begin{equation}\label{eq:c}
 e(\mathcal{C}_n):= \max_{s^n \in S^n}  \, \frac{1}{J_n} \sum^{J_n}_{j=1} \sum_{x^n \in A^n} 
E(x^n| j) W_{s^n}^{n}(D_j^c| x^n) \enspace.  
\end{equation} 
%-------------------------------------random_code-------------------------------------------------------------------------------------------------
%----------------------------------------------------------------------------------------------------------------------------------------------------------
A \textit{correlated random} $(n,J_n,\Gamma,\mu)$ \textit{code} $\mathcal{C}^{\textrm{ran}}_n$ for the
arbitrarily varying wiretap channel is given by a family of wiretap codes
$\{\mathcal{C}_n(\gamma)\}_{\gamma \in \Gamma}$ together with a random experiment choosing $\gamma$
according to a distribution $\mu$ on $\Gamma$. The mean average error probability of a random
$(n.J_n,\Gamma,\mu)$ code $\mathcal{C}^{\textrm{ran}}_n$ is defined analogously to the ordinary one but
with respect to the random experiment choosing $\gamma$ by 
\begin{equation*}
\bar{e} ( \mathcal{C}^{\textrm{ran}}_n  )  :=  
 \max_{s^n \in S^n}  \, \frac{1}{J_n} \sum^{J_n}_{j=1} \sum_{\gamma \in \Gamma} \sum_{x^n \in A^n}  
E^{\gamma} (x^n| j) W_{s^n}^{n}((D^{\gamma}_j)^c| x^n) \mu(\gamma) \enspace.   
\end{equation*} 
%-----------------------------------------------CSI-----------------------------------------------------------------------------------------------------

\if0
If channel state information is available at the transmitter the notion of $(n,J_n)$ code is modified in
that the encoding may depend on the channel index while the decoding sets remain universal,
i.e. independent of the channel index
 $t$. 
The probability of error in \eqref{eq:1} changes to
\begin{equation*}
e_{\textup{CSI}}(\mathcal{C}_n) := \max_{t\in \theta} \, \max_{j \in\mathcal{J}_n} \sum_{x^n\in A^n}
E_t(x^n|j) W_t^{\otimes n}(D_j^c| x^n).
\end{equation*}
\fi

%----------------------------------------------------------------------------------------------------------------------------------------------------------

\begin{definition}\label{code}
A non-negative number $R_S$ is an achievable secrecy rate for the AVWC $\mathfrak{W}$, if there is a
sequence $(\mathcal{C}_n)_{n\in\nn}$ of $(n,J_n)$ codes such that
\[ \lim_{n\to\infty} e(\mathcal{C}_n)=0 
%\textrm{ resp.  } \lim_{n\to\infty} e_{\textup{CSI}}(\mathcal{C}_n)=0
\enspace,\]
\[\liminf_{n\to\infty}\frac{1}{n}\log J_n\ge R_S \enspace, \]
and
\begin{equation}\label{eq:2} 
\lim_{n\to\infty} \max_{s^n \in S^n} I(p_J; V^n_{s^n})=0 \enspace, 
\end{equation}
where $J$ is a uniformly distributed random variable taking values in $\mathcal{J}_n$ and $I(p_J;
V^n_{s^n})$ is the mutual information of $J$ and the output variable $Z^n$ of the eavesdropper's channel
$V^n_{s^n}$. The secrecy capacity then is given as the supremum of all achievable secrecy rates $R_S$ and
is denoted by $C_S(\mathfrak{W})$. 
\end{definition}
Analogously we define the secrecy rates and the secrecy capacity for random codes
$C_{S,\textrm{ran}}(\mathfrak{W})$, if we replace $\mathcal{C}_n$ by $\mathcal{C}^{\textrm{ran}}_n$ in the
above definition. 
\begin{definition}\label{code_R}
A non-negative number $R_S$ is an achievable secrecy rate for correlated random codes for the AVWC
$\mathfrak{W}$, if there is a sequence $(\mathcal{C}^{\emph{ran}}_n)_{n\in\nn}$ of $(n,J_n,\Gamma,\mu)$
codes such that
\[ \lim_{n\to\infty} \bar{e}(\mathcal{C}^{\emph{ran}}_n)=0 
%\textrm{ resp.  } \lim_{n\to\infty} e_{\textup{CSI}}(\mathcal{C}_n)=0
\enspace,\]
\[\liminf_{n\to\infty}\frac{1}{n}\log J_n\ge R_S \enspace, \]
and 
\begin{equation}\label{eq:3} 
\lim_{n\to\infty} \max_{s^n \in S^n} \sum_{\gamma \in \Gamma} I(p_J, V^n_{s^n};\mathcal{C}(\gamma))
\mu(\gamma) =0 \enspace, 
\end{equation}
where $I(p_J,V^n_{s^n};\mathcal{C}(\gamma))$ is the mutual information according to the code
$\mathcal{C}_(\gamma), \ \gamma \in \Gamma$ chosen according to the distribution $\mu$. The secrecy capacity then is given as the supremum of all achievable secrecy rates $R_S$ and
is denoted by $C_{S,\emph{ran}}(\mathfrak{W})$. 
\end{definition}

%--------------------------------------------------------------------------------------------------------------------------------------------------------
%--------------------------------------------------------------------------------------------------------------------------------------------------------
%--------------------------------------------------------------------------------------------------------------------------------------------------------
\section{Capacity Results}

\subsection{Preliminaries}

In what follows we use the notation as well as some properties of \emph{typical} and \emph{conditionally
  typical} sequences from \cite{csis2}. For $p\in\mathcal{P}(A)$, $W:A\to\mathcal{P}(B)$, $x^n\in A^n$,
and $\delta>0$  we denote by $\ty_{p,\delta}^n$ the set of typical sequences and by
$\ty_{W,\delta}^n(x^n) $ the set of conditionally typical sequences given $x^n$ in the sense of \cite{csis2}.\\
The basic properties of these sets that are needed in the sequel are summarised in the following three
lemmata.  
%------------------------------------------------------------------
\begin{lemma}\label{typical}
Fixing $\delta > 0$, for every $p \in \pp(A)$ and  $W:A \to \pp(B)$ we have
\begin{eqnarray*}
p^{\otimes n}(\ty_{p,\delta}^n) & \geq &1- (n+1)^{|A|} 2^{-nc\delta^2} \\
W^{\otimes n}(\ty_{W,\delta}^n(x^n)|x^n) & \geq &1-(n+1)^{|A||B|} 2^{-nc \delta^2}
\end{eqnarray*}
for all $x^n\in A^n$ with $c=1/(2\ln 2)$. In particular, there is $n_0\in\nn$ such that for each
$\delta>0$ and $p\in \pp(A)$, $W:A\to\pp(B)$
\begin{eqnarray*}
 p^{\otimes n}(\ty_{p,\delta}^n) & \geq &1- 2^{-nc'\delta^2} \\
W^{\otimes n}(\ty_{W,\delta}^n(x^n)|x^n) & \geq &1- 2^{-nc' \delta^2}
\end{eqnarray*}
holds with $c'=\frac{c}{2}$.
\end{lemma}
%----------------------------------------------------------------------
\begin{proof}
 Standard Bernstein-Sanov trick using the properties of types from  \cite{csis2} and Pinsker's inequality. 
The details can be found in \cite{wyrem} and references therein for example. \qed
\end{proof}
Recall that for $p\in\pp(A)$ and $W:A\to\pp(B)$, $pW\in\pp(B)$ denotes the output distribution generated
by $p$ and $W$ and that  $x^n \in \ty^n_{p,\delta}$ and $y^n \in \ty^n_{W,\delta}(x^n)$ imply that $y^n
\in \ty^n_{pW,2|A|\delta}$. 
%-----------------------------------------------------------------------------------------------
\begin{lemma}\label{alpha-beta} 
Let $x^n\in \ty^n_{p,\delta}$, then for $V:A\to\pp(C)$
\begin{eqnarray*}
|\ty_{pV,2|A|\delta}^n| &\leq& \alpha^{-1}\\
V^n(z^n|x^n) &\leq& \beta \quad \textrm{for all} \quad z^n \in \ty^n_{V,\delta}(x^n)
\end{eqnarray*} 
hold, where
\begin{eqnarray}
\alpha &=&2^{-n(H(pV)+f_1(\delta))}\label{eq:4}\\
\beta &=&2^{-n(H(V|p)-f_2(\delta))}\label{eq:5}
\end{eqnarray}
with universal $f_1(\delta),f_2(\delta)>0$ satisfying
$\lim_{\delta\to\infty}f_1(\delta)=0=\lim_{\delta\to\infty}f_2(\delta)$. 
\end{lemma}
%-----------------------------------------------------------------------------------------------------
\begin{proof}
Cf. \cite{csis2}.
\end{proof}

%--------------------------------------output------------------------------------------------------------------------------------------
\if0
In addition we need a further lemma which will be used to determine the rates at which reliable
transmission to the legitimate receiver is possible.
%--------------------------------------------------------------------------
\begin{lemma}\label{output-bound}
Let $p, \tilde{p} \in \mathcal{P}(A)$ and two stochastic matrices $W, \widetilde{W}:A \to \mathcal{P}(B)$
be given. Further let $q,\tilde{q} \in \mathcal{P}(B)$ be the output distributions, the former generated
by $p$ and $W$ and the latter by $\tilde{p}$ and $\widetilde{W}$. Fix $\delta \in
(0,\frac{1}{4|A||B|})$. Then for every $n \in \nn$
\begin{equation*}
q^{\otimes n}(\ty^n_{\widetilde{W}, \delta}(\tilde{x}^n)) \leq (n+1)^{|A||B|}
  2^{-n(I(\tilde{p},\widetilde{W})-f(\delta))}
\end{equation*}
for all $\tilde{x}^n \in \ty^n_{\tilde{p},\delta}$ and 
\begin{equation*}
q^{\otimes n}(\ty^n_{W, \delta} (x^n)) \leq (n+1)^{|A||B|} 
  2^{-n(I(p,W)-f(\delta))}  
\end{equation*}
for all $x^n \in \ty^n_{p, \delta}$ holds for a universal $f(\delta) >0$ and $\lim_{\delta\to 0}
f(\delta)=0$.
\end{lemma}
%------------------------------------------------------------------------------
\begin{proof}
Cf. \cite{wyrem}.
\end{proof}
\fi
%-----------------------------------------------------------------------------------------------------------------------------------------
The next lemma is a standard result from large deviation theory.
\begin{lemma}\label{chernoff}(Chernoff bounds)
Let $Z_1,\ldots,Z_L$ be i.i.d. random variables with values in $[0,1]$ and expectation
$\mathbb{E}Z_i=\mu$, and $0<\epsilon<\frac{1}{2}$. Then it follows that
\begin{equation*}
\textrm{Pr} \left\{ \frac{1}{L} \sum^L_{i=1} Z_i \notin [(1\pm\epsilon)\mu]  \right\} \leq 2\exp \left( -L\cdot
\frac{\epsilon^2\mu}{3} \right), \nonumber
\end{equation*}
where $[(1\pm\epsilon)\mu]$ denotes the interval $[(1 - \epsilon)\mu, (1+ \epsilon)\mu]$.
\end{lemma}
%------------------------------------------------------robustification--------------------------------------------------------------------------
For the optimal random coding strategy of the AVWC we need the \textit{robustification technique} by
Ahlswede \cite{ahlsw2} which is formulated as a further lemma. Therefore
let $\Sigma_n$ be the group of permutations acting on $(1,2, \ldots ,n)$. Then every permutation $\sigma
\in \Sigma_n$ induces a bijection $\pi \in \Pi_n$ defined by $\pi: \mathcal{S}^n \to \mathcal{S}^n$ with
$\pi(s^n)=(s_{\sigma(1)}, \ldots ,s_{\sigma(n)})$ for all $s^n=(s_1, \ldots ,s_n) \in \mathcal{S}^n$ and
$\Pi_n$ denotes the group of these bijections.
\begin{lemma}\label{robust}(Robustification technique)
If a function $f: \mathcal{S}^n \to [0,1]$ satisfies
\begin{equation}\label{eq:7}
\sum_{s^n \in \mathcal{S}^n} f(s^n) q(s_1) \cdot \ldots \cdot q(s_n) \geq 1- \gamma
\end{equation}
for all $q \in \mathcal{P}_0 (n,\mathcal{S})$ and some $\gamma \in [0,1]$, then
\begin{equation}\label{eq:8}
\frac{1}{n!} \sum_{\pi \in \Pi_n} f(\pi(s^n)) \geq 1-3\cdot (n+1)^{|\mathcal{S}|} \cdot
\gamma \quad \forall s^n \in \mathcal{S}^n \enspace.
\end{equation}
\end{lemma}
\begin{proof}
The proof is given in \cite{ahlsw2}.
\end{proof}
%-----------------------------------------------symmetrisability----------------------------------------------------------------------------
To reduce the random code for the AVWC $\mathfrak{W}$ to a deterministic code we need the concept of
symmetrisability, which was established for ordinary AVCs in the following  representation by
\cite{erics1}, \cite{csisnara}.
\begin{definition}{\cite{csisnara}}
An AVC is symmetrisable if for some channel $U: A \to S$
\begin{equation}\label{eq:symmable}
\sum_{s \in S} W(y|x,s)U(s|x')=\sum_{s \in S} W(y|x',s)U(s|x)
\end{equation}
for all $x,x' \in A$, $y \in B$.
\end{definition}
A new channel defined by \eqref{eq:symmable} then would be symmetric with respect to all $x, x'\in A$. 
The authors of \cite{csisnara} proved the following theorem which is a concretion of Ahlswede's dichotomy
result for single-user AVC, which states that the deterministic code capacity $C$ is either $C=0$ or
equals the random code capacity. 
\begin{theorem}{\cite{csisnara}}\label{symmetrisable}
$C>0$ if and only if the AVC is non-symmetrisable. If $C>0$, then
\begin{equation}
C=\max_{p \in \mathcal{P}(A)} \min_{W \in \bar{\mathcal{W}}} I(p,W)
\end{equation}
\end{theorem}
Here the RHS gives the random code capacity and $\bar{\mathcal{W}}$ denotes the convex closure of all
channels $W_s$ with $s\in S$, $S$ finite or countable.

\subsection{Random Code Construction}\label{Rancode}

%-----------------------------------------------convex_hull------------------------------------------------------------------------------------------
First let us define the convex hull of the set of channels $\{W_s: s \in S\}$ by the set of channels
$\{W_{q}: q  \in \mathcal{P}(S)\}$, where $W_{q}$ is
defined by
\begin{equation}\label{eq:aver_ch}
W_{q}(y|x)=\sum_{s \in S} W(y|x,s) q(s),
\end{equation}
for all possible distributions $q \in \mathcal{P}(S)$. Accordingly we define $V_q$ and its convex hull $\{V_q:q \in
\mathcal{P}(S)\}$. Then we denote the convex closure of the set of channels $\{(W_s,V_s): s \in S\}$ by
$\overline{\mathfrak{W}}:=\{(W_q,V_q): q \in \mathcal{P}(\tilde{S}), \tilde{S} \subseteq
S, \tilde{S} \ \textrm{is finite}\}$. Occasionally, we restrict $q$ to be from the set of all types
$\mathcal{P}_0(n,S)$ of state sequences $s^n \in S^n$.
\begin{lemma}\label{cap_closure}
The secrecy capacity $C_S(\mathfrak{W})$ of the arbitrarily varying wiretap channel AVWC $\mathfrak{W}$
equals the secrecy capacity of the arbitrarily varying wiretap channel $\overline{\mathfrak{W}}$.  
\end{lemma}
\begin{proof}
The proof was given for an ordinary arbitrarily varying channel AVC without secrecy criterion in
\cite{csis2} and for an AVWC under the weak secrecy criterion in \cite{molav}. Let $\tilde{W}_1, \ldots, \tilde{W}_n$ be averaged channels as defined in \eqref{eq:aver_ch} and
a channel $W^n_{\tilde{q}} : A^n \to \mathcal{P}(B^n)$ with $\tilde{q}=\prod^n_{i=1} q_i$, $\tilde{q} \in
\mathcal{P}(S^n)$, $q_i \in \mathcal{P}(S)$ defined by
\begin{equation*}
W^n_{\tilde{q}}(y^n|x^n) =\prod^n_{i=1} \tilde{W}_i (y_i|x_i)  = \prod^n_{i=1} W_{{q}_i}(y_i|x_i) 
 =\sum_{s^n \in S^n} W^{n}(y^n|x^n,s^n) \tilde{q}(s^n)
\end{equation*} 
If we now use the same $(n,J_n)$ code $\mathcal{C}_n$ defined by the same pair of encoder and decoding sets
as for the AVWC $\mathfrak{W}$ the error probability for transmission of a single codeword by the channel
$W^n_{\tilde{q}}$ is given by
\begin{equation*}
\sum_{x^n \in A^n} E(x^n| j)  W^n_{\tilde{q}}(D_j^c| x^n)
 =\sum_{s^n \in S^n} \tilde{q}(s^n) \sum_{x^n \in A^n} E(x^n| j) W_{s^n}^{n}(D_j^c| x^n)
\end{equation*}
and we can bound the average error probability by
\begin{equation*}
\begin{split}
&\frac{1}{J_n} \sum^{J_n}_{j=1} \sum_{x^n \in A^n} 
E(x^n| j) W^n_{\tilde{q}} (D_j^c| x^n)\\
& \leq  \max_{s^n \in S^n}  \,  \frac{1}{J_n} \sum^{J_n}_{j=1} \sum_{x^n \in A^n} 
E(x^n| j) W_{s^n}^{n}(D_j^c| x^n)=e(\mathcal{C}_n) \enspace.  
\end{split}
\end{equation*} 
Otherwise, because $\mathfrak{W}$ is a subset of $\overline{\mathfrak{W}}^n$, which is the closure of
the set of channels $(W^n_{\tilde{q}},V^n_{\tilde{q}})$, the opposite inequality holds for the channel
$W^n_{\tilde{q}}$ that maximizes the error probability. Because $V^n_{\tilde{q}}$ is defined analogously
to $W^n_{\tilde{q}}$, we can define for the $(n,J_n)$ code
\begin{equation}
\hat{V}(z^n|j):=\sum_{x^n \in A^n} E(x^n| j)  V^n_{\tilde{q}}(z^n| x^n)
\end{equation}
for all $z^n \in C^n$, $j \in \mathcal{J}_n$. Then
\begin{equation}\label{eq:chann_J}
\hat{V}(z^n|j)  =\sum_{s^n \in S^n} \tilde{q}(s^n) \sum_{x^n \in A^n} E(x^n| j)  V^n_{{s^n}}(z^n|x^n)
 =\sum_{s^n \in S^n} \tilde{q}(s^n) \hat{V}^n_{{s^n}}(z^n|j)
\end{equation}
and because of the convexity of the mutual information in the channel $\hat{V}$ and \eqref{eq:chann_J} it
holds that
\begin{equation}
I(J,Z^n_{\tilde{q}}) \leq \sum_{s^n \in S^n} \tilde{q}(s^n) I(J;Z^n_{s^n}) \leq \sup_{s^n} I(J,Z^n_{s^n}).
\end{equation}
Now because $\{\hat{V}^n_{{s^n}}(z^n|j): s^n \in S^n\}$ is a subset of $\{\hat{V}(z^n|j): \tilde{q} \in
\mathcal{P}(S^n)\}$ we end in
\begin{equation*}
\sup_{\tilde{q} \in \mathcal{P}(S^n)} I(J,Z^n_{\tilde{q}})=  \sup_{s^n} I(J,Z^n_{s^n}) \enspace.
\end{equation*}
\qed \end{proof} 
Now we can proceed in the construction of the random code of the AVWC $\mathfrak{W}$. 
\begin{definition}\label{best}
We call a channel to the eavesdropper a best channel if there exist a channel $V_{q^*} \in  \{V_q:q \in
\mathcal{P}(S)\}$ such that all other channels from $\{V_q:q \in \mathcal{P}(S)\}$ are degraded versions
of $V_{q^*}$. If we denote the output of any channel $V_q$, $q\in \mathcal{P}(S)$ by $Z_q$ it holds
that 
\begin{equation}\label{eq:degr_best}
X \to Z_{q^*} \to Z_q, \quad \forall q \in \mathcal{P}(S).
\end{equation}
\end{definition}
%--------------------------------------random-code-capacity------------------------------------------------------------------------------------
\begin{proposition}\label{random_code}
Provided that there exist a best channel to the eavesdropper, for the random code secrecy capacity
$C_{S,\textrm{ran}}(\mathfrak{W})$ of the AVWC $\mathfrak{W}$ it holds that 
\begin{equation}\label{eq:rand_cap}
C_{S,\textrm{ran}}(\mathfrak{W}) \geq \max_{p \in \mathcal{P}(A)} ( \min_{q \in \mathcal{P}(S)}
I(p,W_{q})-\max_{q \in \mathcal{P}(S)}I(p,V_{q})).
\end{equation}
\end{proposition}
\begin{proof}
The proof is based on Ahlswedes \textit{robustification technique} \cite{ahlsw2} and is divided in two parts:\\
\textit{step 1} ): 
The set
\begin{equation*}
\overline{\mathcal{W}} := \{ (W^{n}_q, V^{ n}_q): q \in \mathcal{P}(S) \}
\end{equation*}
corresponds to a compound wiretap channel indexed by the set of all possible distributions $q \in
\mathcal{P}(S)$ on the set of states $S$. First we show, that there exist a deterministic code for
the compound wiretap channel $\overline{\mathcal{W}}$ that achieves the lower bound on the random code
secrecy capacity of the AVWC $\mathfrak{W}$ given in \eqref{eq:rand_cap}.  

In \cite{bjela2} it was shown that for a compound wiretap channel $\{(W_t,V_t): t \in \theta\}$ without
channel state information at the legitimate receivers the secrecy capacity is bounded from below by
\begin{equation}\label{eq:comp_cap}
C_{S,\textrm{comp}} \geq \max_{p \in \mathcal{P}(A)} ( \min_{t \in \theta}I(p,W_{s})-\max_{t \in
  \theta}I(p,V_{s})). 
\end{equation}
In accordance with the proof of \eqref{eq:comp_cap} in \cite{bjela2} we define a set  of i.i.d. random variables
$\{X_{jl}\}_{j \in [J_n], l \in [L_n]}$ each according to the distribution $p' \in \mathcal{P}(A^n)$ with
\begin{equation}\label{eq:pruned}
p'(x^n):= \left \{ \begin{array}{ll}
\frac{p^{\otimes n}(x^n)}{p^{\otimes n}(\ty^n_{p,\delta})} & \textrm{if $x^n \in \ty^n_{p,\delta}$},\\
0 &  \textrm{otherwise},
\end{array} \right.
\end{equation}
for any $p \in \mathcal{P}(A)$, and where $J_n$ and $L_n$ are chosen as
\begin{eqnarray}
J_n&=& \lfloor 2^{n[\inf_{q \in \mathcal{P}(S)} I(p,W_{q})-\sup_{q \in \mathcal{P}(S)} I(p,V_{q})
  -\tau]} \rfloor \label{eq:14}\\
L_{n}&=& \lfloor 2^{n[\sup_{q \in \mathcal{P}(S)} I(p,V_{q})+ \frac{\tau}{4}] } \rfloor \label{eq:15}
\end{eqnarray}
with $\tau>0$. Now we assume that there exist a best channel to the eavesdropper $V_{q^*}$ in contrast to
the proof in \cite{bjela2}. Hence by the definition of $V_{q^*}$ in \eqref{eq:degr_best} and because
the mutual Information $I(p,V)$ is convex in $V$ and every member of $\{V_q\}_{q \in \mathcal{P}(S)}$ 
is a convex combination of the set $\{V_s\}_{s \in S}$, it holds that 
\begin{equation}\label{eq:best_channel}
\begin{split}
I(p,V_{q^*})=\sup_s I(p,V_s)=\sup_{q \in \mathcal{P}(S)} I(p,V_{q})
\end{split}
\end{equation}
for all $p\in \mathcal{P}(A)$. Note that because of \eqref{eq:best_channel} for $|S| < \infty$ $V_{q^*}
\in \{V_{s}: s  \in S\}$, which means that $q^*$ is a one-point distribution. 
 
By the definition of the compound channel $\overline{\mathcal{W}}$ the channels to the eavesdropper are
of the form
\begin{equation}
V^{n}_q (z^n|x^n) :=\prod^n_{i=1} V_{q}(z_i|x_i)
\end{equation}
for all $q \in \mathcal{P}(S)$. 
Then following the same approach as in the proof in \cite{bjela2} we define
\begin{equation*}
\tilde{Q}_{q,x^n}(z^n)=V_q^{n}(z^n|x^n)\cdot\mathbf{1}_{\mathcal{T}^n_{V_q,\delta}(x^n)}(z^n),
\end{equation*}
and 
\begin{equation}\label{eq:12}
\Theta'_q(z^n) =\sum_{x^n \in \mathcal{T}^n_{p,\delta}} p'(x^n)\tilde{Q}_{q,x^n}(z^n).
\end{equation}
for all $z^n\in C^n$.
Now let $\mathcal{B}:=\{z^n \in C^n : \Theta'_q(z^n) \geq \epsilon\alpha_q\}$ where
$\epsilon=2^{-nc'\delta^2}$ (cf. Lemma \ref{typical}) and $\alpha_q$ is from (\ref{eq:4}) in
Lemma \ref{alpha-beta} computed with respect to $p$ and $V_q$. By Lemma \ref{alpha-beta} the support of
$\Theta'_q$ has cardinality $\leq \alpha^{-1}_q$ since for each $x^n \in \mathcal{T}^n_{p,\delta}$ it holds
that $\mathcal{T}^n_{V_q,\delta}(x^n) \subset \mathcal{T}^n_{pV_q, 2|A|\delta}$, which implies that
$\sum_{z^n \in \mathcal{B}} \Theta_q(z^n) \geq 1-2\epsilon$, if 
\begin{eqnarray}
\Theta_q(z^n)&=&\Theta'_q(z^n)\cdot\mathbf{1}_{\mathcal{B}}(z^n) \quad \textrm{and} \nonumber\\
Q_{q,x^n}(z^n)&=&\tilde{Q}_{q,x^n}(z^n) \cdot \mathbf{1}_{\mathcal{B}}(z^n).\label{eq:13}
\end{eqnarray}
Now it is obvious from (\ref{eq:12}) and the definition of the set $\mathcal{B}$ that for any $z^n\in \mathcal{B}$
$\Theta_q(z^n)=\mathbb{E}Q_{q,X_{jl}}(z^n)\ge \epsilon \alpha_q$ if $\mathbb{E}$ is the expectation
value with respect to the distribution $p'$. Let $\beta_q$ defined as in \eqref{eq:5} with respect to
$V_q$. For the random variables $\beta^{-1}_q
Q_{q,X_{jl}}(z^n)$ define the event
\begin{equation}\label{eq:17}
\iota_j(q)=\bigcap_{z^n \in C^n} \left\{\frac{1}{L_{n}}\sum_{l=1}^{L_{n}} Q_{q,X_{jl}}(z^n) \in [(1 \pm
  \epsilon) \Theta_q(z^n)]\right \},
\end{equation}
and keeping in mind that  $\Theta_q(z^n) \geq \epsilon \alpha_q$ for all $z^n \in \mathcal{B}$. Then it
follows that for all $j \in [J_n]$ and for all $s\in S$
\begin{equation}\label{eq:18}
\textrm{Pr}\{ (\iota_j(q))^c\} \leq 2 |C|^n
\exp  \Big(- L_{n} \frac{ 2^{-n[I(p,V_q)+g(\delta)]}}{3}   \Big)
\end{equation}
by Lemma \ref{chernoff}, Lemma \ref{alpha-beta}, and our choice $\epsilon=2^{-nc'\delta^2}$ with 
$g(\delta):=f_1(\delta)+f_2(\delta)+3c'\delta^2$. Making $\delta>0$ sufficiently small we have for all 
sufficiently large $n\in \mathbb{N}$
\[ L_{n} 2^{-n[I(p,V_q)+g(\delta)]}\ge 2^{n\frac{\tau}{8}}. \]
Thus, for this choice of $\delta$ the RHS of (\ref{eq:18}) is double exponential in $n$ uniformly in
$q \in \mathcal{P}(S)$  and can be made smaller than $\epsilon J_n^{-1}$ for all $j \in [J_n]$ and all
sufficiently large $n \in \mathbb{N}$. I.e. 
\begin{equation}\label{eq:19}
 \textrm{Pr}\{ (\iota_j(q))^c\} \leq \epsilon J_n^{-1} \quad \forall q \in \mathcal{P}(S)
\end{equation} 

Now we will show that we can achieve reliable transmission to the legitimate receiver governed by
$\{(W^{n}_q:  q \in \mathcal{P}(S) \}$ for all messages $j\in [J_n]$ when randomising over
the index  $l \in L_n$ but without the need of decoding $l \in [L_n]$. To this end define
$\mathcal{X}=\{X_{jl}\}_{j \in [J_n], l \in [L_n]}$ to be the set of random variables with $X_{jl}$ are
i.i.d. according to $p'$ defined in \eqref{eq:pruned}. Define now the random decoder
$\{D_{j}(\mathcal{X}) \}_{j \in [J_n]} \subseteq B^n$ analogously as in \cite{bjela2}, \cite{bjela3}.
%--------------------------------------------------------------------------------------------------------------------------------------
\if0
by
\begin{equation}\label{eq:achiev-is-5}
 D_j(\mathcal{X}):=D'_j(\mathcal{X})\cap \Big( \bigcup_{\substack{j'\in [J_n] \\ j'\neq j}}
   D'_{j'}(\mathcal{X}) \Big)^c. 
\end{equation} 
with
\begin{equation}\label{eq:achiev-is-4}
 D'_j(\mathcal{X}):=\bigcup_{q' \in \mathcal{P}(S)}\bigcup_{k\in[L_{n}]} \ty_{W_{q'},\delta}^{n}(X_{jk}),
\end{equation} 
and consequently the random average probabilities of error for a specific channel $W_q$, $q \in
\mathcal{P}(S)$ by 
\begin{equation}\label{eq:error-t}
\lambda_n^{(q)}(\mathcal{X}):=\frac{1}{J_n}\sum_{j\in [J_n]}\frac{1}{L_{n}}\sum_{l\in[L_{n,t}]} 
  W_t^{\otimes n}((D_j(\mathcal{X}) )^{c}| X_{jl} ) .
\end{equation}
\fi
%----------------------------------------------------------------------------------------------------------------------------------------
Then it was shown by the authors, that there exist a sequence of $(n,J_n)$ codes for the compound wiretap
channel in the particular case without CSI  with arbitrarily small mean average error 
\begin{equation*}
\mathbb{E}_{\mathcal{X}}(\lambda^{(q)}_n(\mathcal{X})) \leq 2^{-na}
\end{equation*}
for all $q \in \mathcal{P}(S)$ and sufficiently large $n \in \mathbb{N}$. Additionally we define for each
$q \in \mathcal{P}(S)$
\begin{equation}\label{eq:20}
\iota_0(q)=\{\lambda^{(q)}_n(\mathcal{X})) \leq 2^{-n\frac{a}{2}}\}
\end{equation}
and set
\begin{equation}\label{eq:21}
\iota:=\bigcap_{q\in \mathcal{P}_0(n,S)}\bigcap_{j=0}^{J_n}\iota_j(q)
\end{equation}
Then with \eqref{eq:19}, \eqref{eq:20} and applying the union bound we obtain 
\begin{equation*}
\textrm{Pr}\{ \iota^c\} 
%\leq \sum_{q \in \mathcal{P}_0(n,S)}\sum_{j=1}^{J_n}\textrm{Pr}\{(\iota_j(q))^c\} 
\leq 2^{-nc}
\end{equation*}
for a suitable positive constant $c>0$ and all sufficiently large $n\in\mathbb{N}$ (Cf. \cite{bjela2}).\\
Hence, we have shown that there exist realisations $\{x_{jl}\}$ of $\{X^n_{jl}\}_{j \in [J_n], l \in
  [L_n]}$ such that $x_{jl} \in \iota$ for all $j \in [J_n]$ and $l \in [L_n]$. 
Now following the same argumentation as in \cite{bjela2}, \cite{bjela3} we obtain that there is a
sequence of $(n,J_n)$ codes that for all codewords $\{x_{jl}\}$ it follows by construction that
\begin{equation}\label{eq:22}
\frac{1}{J_n}\sum_{j \in [J_{n}]} \frac{1}{L_n}\sum_{l \in[L_n]}W^{n}_{q}(D_{j}^c|x_{jl})\le 2^{-n a'}
\end{equation} 
is fulfilled for $n \in \mathbb{N}$ sufficiently large and for all $q \in \mathcal{P}(S)$ with
$a'>0$. So we have found a $(n,J_n)$ code with average error probability upper bounded by
\eqref{eq:22}. Further, for the given code and a random variable $J$ uniformly
distributed on the message set $\{1, \ldots, J_n\}$ it holds that  
\begin{equation}\label{eq:23}
I(p_J;V^{n}_q) \leq \epsilon'
\end{equation}
uniformly in $q \in \mathcal{P}(S)$. Both \eqref{eq:22} and \eqref{eq:23} ensure that in the scenario of
the compound wiretap channel $\overline{\mathcal{W}}$ the legitimate receiver can identify each message
$j$ from the message set $\{1, \ldots, J_n\}$ with high probability, while at the same time the
eavesdropper receives almost no information about it. That is, that all numbers $R_S$ with
\begin{equation}\label{eq:24}
R_S \leq \inf_{q \in \mathcal{P}(S)} I(p,W_{q})-\sup_{q \in \mathcal{P}(S)}I(p,V_{q})
\end{equation}
are achievable secrecy rates of the compound wiretap channel $\overline{\mathcal{W}}$.\\
\textit{step 2} ): \textit{Robustification} :
In the second step we derive from the deterministic $(n,J_n)$ code for the above mentioned compound
wiretap channel $\overline{\mathcal{W}}$ a $(n,J_n)$ random code $\mathcal{C}^{\textrm{ran}}_n$ for the
AVWC $\mathfrak{W}$, which achieves the same secrecy rates.  
We note first that by \eqref{eq:best_channel} and \eqref{eq:23}
\begin{equation}\label{eq:25}
\max_{s^n \in S^n} I(p_J,V_{s^n})=I(p_J,V^{n}_{q^*}) \leq \epsilon',
\end{equation}
which means, that, due to the assumption of a best channel to the eavesdropper, the code achieving the
secrecy rate for the best channel to the eavesdropper fulfills the secrecy criterion for a channel with
any state sequence $s^n \in S^n$. 
Now, as already mentioned we use the robustification technique (cf. Lemma \ref{robust}) to derive from the
deterministic code $\mathcal{C}_{\overline{\mathcal{W}}}= \{x_{jl}, D_{j}: j \in [J_n], l \in [L_n] \}$
of the compound wiretap channel $\overline{\mathcal{W}}$ the random code for the AVWC
$\mathfrak{W}$. Therefore, for now let $S$ to be finite. With \eqref{eq:22} it holds that
\begin{equation}
\frac{1}{J_n}\sum_{j \in [J_{n}]} \frac{1}{L_n}\sum_{l \in[L_n]} \sum_{s^n \in S^n} 
W^{n} (D_{j}|x_{jl},s^n)  q^{\otimes n}(s^n)  \geq 1 - 2^{-n a'}
\end{equation}
for all $q^{\otimes n}= \prod^n_{i=1}q$ and in particular for all $q \in \mathcal{P}_0(n,S)$. Now let
$\pi \in \Pi_n$ be the bijection on $S^n$ induced by the permutation $\sigma \in \Sigma_n$. Since
\eqref{eq:7} is fulfilled with 
\begin{equation}
f(s^n)=\frac{1}{J_n}\sum_{j \in [J_{n}]} \frac{1}{L_n}\sum_{l \in[L_n]} 
W^{ n} (D_{j}|x_{jl},s^n)
\end{equation}
it follows from \eqref{eq:8} that
\begin{equation}\label{eq:28}
\frac{1}{n!}\sum_{\pi \in \Pi_n} \frac{1}{J_n}\sum_{j \in [J_{n}]} \frac{1}{L_n}\sum_{l \in[L_n]}  
W^{n} (D_{j}|x_{jl},\pi(s^n))  
\geq 1 - (n+1)^{|S|}2^{-n a'}
\end{equation}
for all $s^n \in S^n$. Hence by defining $\mathcal{C}^{\pi} := \{ \pi^{-1}(x^n_{jl}), \pi^{-1}(D_{j}) \}$
as a member of a family of codes $\{\mathcal{C}^{\pi}\}_{\pi \in  \Pi_n}$ together with a random variable
$K$ distributed according to $\mu$ as the uniform distribution on $\Pi_n$, \eqref{eq:28} is equivalent to
\begin{equation}\label{eq:29}
\mathbb{E}_{\mu}(\bar{\lambda}_n(\mathcal{C}^{K}, W^n_{s^n})) \leq (n+1)^{|S|}2^{-na'}=:\lambda_n
\end{equation}
with $\bar{\lambda}_n(\mathcal{C}^{\pi},W^n_{s^n})$ as the respective average error probability for
$K=\pi$ and it
holds for all $s^n \in S^n$. Thus we have shown that 
\begin{equation}\label{eq:30}
\mathcal{C}^{\textrm{ran}}_n :=\{ (  \pi^{-1}(x_{jl}), \pi^{-1}(D_{j}) ): j \in [J_n], l \in [L_n], \pi \in \Pi_n,
\mu \}
\end{equation}
is a $(n,J_n,\Pi_n,\mu)$ random code for the AVC channel $\mathcal{W}^n=\{W_{s^n}:s^n \in S^n\}$ with the
mean average error probability $\mathbb{E}_{\mu}(\bar{\lambda}_n(\mathcal{C}^{K}, W^n_{s^n}))$ upper bounded by
$\lambda_n$ as in
\eqref{eq:29}. 

Now it is easily seen that
\begin{equation}\label{eq:31}
p^{\mathcal{C}^{\pi}}_{JZ^n_{q^*}}(j,z^n)  =
\frac{1}{J_n}\frac{1}{L_n} \sum^{L_n}_{l=1} V^{ n}_{q^*} (\pi^{-1}(z^n)| \pi^{-1}(x_{jl}))
 =p_{JZ^n_{q^*}}.
\end{equation}
Actually, it still holds that
\begin{equation}
p^{\mathcal{C}^{r}}_{JZ^n_{q^*}}(j,z^n) 
= \frac{1}{n!}\sum_{\pi \in \Pi_n} p^{\mathcal{C}^{\pi}}_{JZ^n_{q^*}}(j,z^n) = p_{JZ^n_{q^*}} \enspace.
\end{equation}
With \eqref{eq:31} and the representation of the mutual information by the information divergence we
obtain from \eqref{eq:25}
\begin{equation}
\begin{split}
\mathbb{E}_{\mu}  (D(p^{\mathcal{C}^{K}}_{JZ^n_{q^*}}  || p_J \otimes
p^{\mathcal{C}^{K}}_{Z^n_{q^*}}))
& =\frac{1}{n!} \sum_{\pi \in \Pi_n} D(p^{\mathcal{C}^{\pi}}_{JZ^n_{q^*}} || p_J \otimes
p^{\mathcal{C}^{\pi}}_{Z^n_{q^*}}) \\
& =\frac{1}{n!} \sum_{\pi \in \Pi_n} D(p_{JZ^n_{q^*}} || p_J \otimes p_{Z^n_{q^*}})
   =I(p_J , V^{n}_{q^*}) \leq \epsilon' \enspace.
\end{split}
\end{equation}
Thus we have constructed a random $(n,J_n,\Gamma,\mu)$ code $\mathcal{C}^{\textrm{ran}}_n$ with mean
average error probability bounded for all $s^n \in S^n$ as in \eqref{eq:29} and which fulfills the strong
secrecy criterion almost surely, provided that there exist a best channel to the eavesdropper. By the
construction of the random code it follows that the secrecy rates given by \eqref{eq:24} for the compound
wiretap channel $\overline{\mathcal{W}}$ achieved by the deterministic code
$\mathcal{C}_{\overline{\mathcal{W}}}$ are achievable secrecy rates for the AVWC $\mathfrak{W}$ with
random code $\mathcal{C}^{\textrm{ran}}_n$. That is, we have shown that all rates $R_S$ with
\begin{equation}\label{eq:34}
R_S \leq \max_{p\in \mathcal{P}(A)} ( \min_{q \in \mathcal{P}(S)} I(p,W_{q})-\max_{q \in \mathcal{P}(S)} 
I(p,V_{q})) \enspace.
\end{equation}
are achievable secrecy rates of the arbitrarily varying wiretap channel AVWC with random code
$\mathcal{C}^{\textrm{ran}}_n$. \qed
\end{proof}

\subsection{Deterministic Code Construction}\label{determ_code}

Because the code $\mathcal{C}^{\pi}$ that is used for the transmission of a single message is subjected
to a random selection, reliable transmission can only be guaranteed if the outcome of the random experiment
can be shared by both the transmitter and the receiver. One way to inform the receiver about the code
that is chosen is to add a short prefix to the actual codeword. Provided that the number of codes is
small enough, the transmission  of these additional prefixes causes no essential loss in rate. In the
following we use the \textit{elimination technique} by Ahlswede  \cite{ahlsw3} which has introduced the
above approach to derive deterministic codes from random codes for determining capacity of arbitrarily
varying channels. Temporarily we drop the requirement of a best channel to the eavesdropper and state the
following theorem. 
\begin{theorem}\label{achievability}
\begin{enumerate}
\item
Assume that for the AVWC $\mathfrak{W}$ it holds that $C_{S,\mathrm{ran}}(\mathfrak{W})>0$. Then
the secrecy capacity $C_S(\mathfrak{W})$ 
%of the AVWC $\mathfrak{W}$ 
equals its random code secrecy capacity $C_{S,\mathrm{ran}}(\mathfrak{W})$, 
\begin{equation}
C_S(\mathfrak{W}) = C_{S,\mathrm{ran}}(\mathfrak{W}), 
\end{equation}
if and only if the channel to the legitimate receiver is non-symmetrisable.
\item
% If the channel to the legitimate receiver is non-symmetrisable it always holds that
% \begin{equation*}
% C_S(\mathfrak{W}) = C_{S,\mathrm{ran}}(\mathfrak{W}). 
% \end{equation*}
If $C_{S,\mathrm{ran}}(\mathfrak{W})=0$ it always holds that $C_S(\mathfrak{W}) =0$.
\end{enumerate}
\end{theorem}
First, if the channel to the legitimate receiver is symmetrisable then the deterministic
code capacity of the channel to the legitimate receiver equals zero by Theorem \ref{symmetrisable}
and no reliable transmission of messages is possible. Hence the deterministic code secrecy capacity of
the arbitrarily varying wiretap channel also equals zero although the random code secrecy capacity could
be greater than zero. So we can restrict to the case in which the channel to the
legitimate receiver is non-symmetrisable. If $C_S(\mathfrak{W}) = C_{S,\textrm{ran}}(\mathfrak{W})>0$, then
the channel to the legitimate receiver must be nonsymmetrisable. For the other direction, because the
secrecy capacity of the AVWC $\mathfrak{W}$ cannot be greater than the random code secrecy capacity it
suffices to show that $C(\{W_{s^n}\})>0$ implies that $C_S(\mathfrak{W}) \geq
C_{S,\textrm{ran}}(\mathfrak{W})$. Here $C(\{W_{s^n}\})$ denotes the capacity of the arbitrarily varying
channels to the legitimate receiver without secrecy. The proof is given in the two paragraphs
\textit{Random code reduction} and \textit{Elimination of randomness}.  
\subsubsection{Random Code Reduction}
We first reduce the random code
$\mathcal{C}^{\textrm{ran}}$  to a new random code selecting only a small number of deterministic codes
from the former, and averaging over this codes gives a new random code with a constant small mean average
error probability, which additionally fulfills the secrecy criterion. 
\begin{lemma}\label{reduct}(Random Code Reduction)
Let $\mathcal{C}(\mathcal{Z})$ be a random code for the AVWC $\overline{\mathfrak{W}}$ consisting of a
family $\{\mathcal{C}(\gamma)\}_{\gamma \in \Gamma}$ of wiretap codes where $\gamma$ is chosen according
to the distribution $\mu$ of $\mathcal{Z}$. Then let
\begin{equation}
\bar{e} (\mathcal{C}^{\mathrm{ran}}_n)=\max_{s^n}\mathbb{E}_{\mu}{e}(s^n|\mathcal{C}(\mathcal{Z})) \leq \lambda_n \quad
\textrm{and}, \quad 
\max_{s^n}\mathbb{E}_{\mu} I(p_J,V_{s^n};\mathcal{C}(\mathcal{Z})) \leq \epsilon'_n \enspace.
\end{equation}
Then for any $\epsilon$ and $K$ satisfying 
\begin{equation}\label{eq:assumpt}
\epsilon > 4 \max\{\lambda_n,\epsilon'_n\} \quad \textrm{and} \quad K > \frac{2n\log|A|}{\epsilon} (1+n\log|S|)
\end{equation}
there exist $K$ deterministic codes $\mathcal{C}_i$, $i=1, \ldots ,K$ chosen from the random code by random
selection such that
\begin{equation}\label{eq:crit_reduced}
\frac{1}{K} \sum^K_{i=1} {e}(s^n|\mathcal{C}_i) \leq \epsilon \quad \textrm{and} \quad \frac{1}{K}
\sum_{i=1}^K I(p_J,V_{s^n};\mathcal{C}_i) \leq \epsilon
\end{equation}
for all $s^n \in S^n$. 
\end{lemma}
\begin{proof}
The proof is analogue  to the proof of Lemma $6.8$ \cite{csis2}, where a similar assertion
in terms of the maximal probability of error for single user AVCs without secrecy criterion is
established. Cf. also \cite{ahlsw3}. Let $\mathcal{Z}$ be the random variable distributed according to
$\mu$ on $\Gamma$ for the $(n,J_n,\Gamma, \mu)$ random code. Now consider $K$ independent repetitions of
the random experiment of code selections according to $\mu$ and call the according random variables
$\mathcal{Z}_i$, $i \in \{1,\ldots, K\}$. Then for any $s^n \in S^n$ it holds that
\begin{equation*}
\begin{split}
\textrm{Pr}\Big\{ \frac{1}{K} \sum^K_{i=1} e(s^n|\mathcal{C}(\mathcal{Z}_i))  \geq  \epsilon & \quad
\textrm{or} \quad \frac{1}{K}  \sum_{i=1}^K   I(p_J,V_{s^n};\mathcal{C}(\mathcal{Z}_i))  \geq \epsilon \Big\} \\
\leq \textrm{Pr}\Big\{\exp  \sum^K_{i=1} \frac{e(s^n|\mathcal{C}(\mathcal{Z}_i))}{n \log|A| } & \geq 
\exp  \frac{K \epsilon}{n\log|A|} \Big\}\\
+  \textrm{Pr} & \Big\{ \exp   \sum_{i=1}^K  \frac{I(p_J,V_{s^n};  \mathcal{C}(\mathcal{Z}_i))}{n\log|A|} \geq
\exp \frac{K\epsilon}{n\log|A|} \Big\}, 
\end{split}
\end{equation*}
and by Markov's inequality
\begin{equation*}
\begin{split}
\textrm{Pr}\Big\{ \frac{1}{K} \sum^K_{i=1} e(s^n|\mathcal{C}(\mathcal{Z}_i))  \geq  \epsilon & \quad
\textrm{or} \quad \frac{1}{K}  \sum_{i=1}^K   I(p_J,V_{s^n};\mathcal{C}(\mathcal{Z}_i))  \geq \epsilon \Big\} \\
\leq \exp \Big(- \frac{K\epsilon}{n\log|A|} \Big) \mathbb{E} \exp & \sum^K_{i=1}   \frac{e(s^n|
  \mathcal{C}(\mathcal{Z}_i))}{n\log|A|}  \\
+  \exp &  \Big(-\frac{K\epsilon}{n\log|A|} \Big)  \mathbb{E} \exp \sum^K_{i=1} \frac{I(p_J,V_{s^n};
\mathcal{C}(\mathcal{Z}_i))}{n\log|A|} \enspace. 
\end{split}
\end{equation*}
Now because of the independency of the random variables $\mathcal{Z}_i$ and because all $\mathcal{Z}_i$ are
distributed as $\mathcal{Z}$ and we have $\exp t \leq 1+t$, for $0\leq t \leq1$ ($\exp$ to  the base
$2$), we can give the following upper bounds 
\begin{equation}
\Big(\mathbb{E} \exp \frac{e(s^n|  \mathcal{C}(\mathcal{Z}))}{n\log|A|}\Big)^K \leq 
\Big(1+\mathbb{E} \frac{e(s^n| \mathcal{C}(\mathcal{Z}))}{n\log|A|} \Big)^K
\leq \Big(1+ \frac{\lambda_n}{n\log|A|} \Big)^K
\end{equation} 
and
\begin{equation}
\Big(\mathbb{E} \exp \frac{I(p_J,V_{s^n}; \mathcal{C}(\mathcal{Z}))}{n\log|A|} \Big)^K  \leq 
\Big(1+\mathbb{E} \frac{I(p_J,V_{s^n}; \mathcal{C}(\mathcal{Z}))}{n\log|A|} \Big)^K
\leq \Big(1+ \frac{\epsilon'_n}{n\log|A|} \Big)^K \enspace.
\end{equation} 
Hence we obtain for any $s^n \in S^n$
\begin{equation*}
\begin{split}
\textrm{Pr} \Big\{ \frac{1}{K} \sum^K_{i=1}  e(s^n|\mathcal{C}(\mathcal{Z}_i)) \geq  \epsilon & \quad
\textrm{or} \quad \frac{1}{K}  \sum_{i=1}^K   I(p_J,V_{s^n};\mathcal{C}(\mathcal{Z}_i)) \geq \epsilon\} \\
\leq \exp \Big[- K\Big (\frac{\epsilon}{n\log|A|} - & \log(1+\frac{\lambda_n}{n\log|A|})\Big)\Big] \\
& +  \exp \Big[\Big(-K(\frac{\epsilon}{n\log|A|} - \log(1+\frac{\epsilon'_n}{n\log|A|})\Big)\Big]\\
\leq 2\exp\Big[  - K\Big(\frac{\epsilon}{n\log|A|} 
 &-\log(1+\max\{\frac{\lambda_n}{n\log|A|},\frac{\epsilon'_n}{n\log|A|}\})\Big)\Big] \enspace.
\end{split}
\end{equation*}
Then
\begin{equation}\label{eq:probability}
\begin{split}
\textrm{Pr}\Big\{ \frac{1}{K} \sum^K_{i=1} e(s^n|  \mathcal{C}  (\mathcal{Z}_i)) 
 \leq  \epsilon  \ \textrm{and} 
\ \frac{1}{K}  \sum_{i=1}^K   I(p_J,V_{s^n};\mathcal{C}(\mathcal{Z}_i)) \leq \epsilon,
\forall s^n \in S^n \Big\} \\
\geq 1-2|S|^n \exp\Big[  - K \Big(\frac{\epsilon}{n\log|A|}  
 -\log(1+\max\{\frac{\lambda_n}{n\log|A|},\frac{\epsilon'_n}{n\log|A|}\})\Big)\Big],
\end{split}
\end{equation}
which is strictly positive, if we choose
\begin{equation*}
\epsilon \geq 2 n\log|A| \log(1+\max\{\frac{\lambda_n}{n\log|A|},\frac{\epsilon'_n}{n\log|A|}\})
\end{equation*}
and
\begin{equation}\label{eq:bound}
K \geq \frac{2 \log|A|}{\epsilon}(n+n^2\log|S|) \enspace.
\end{equation}
Now because for $0 \leq t \leq1$ and $\log$ to the base $2$ it holds that
\begin{equation*}
t \leq \log(1+t) \leq 2t \enspace,
\end{equation*}
we increase the lower bound for choosing $\epsilon$ if
\begin{equation*}
\epsilon \geq 4\max\{\lambda_n, \epsilon'_n\} \enspace.
\end{equation*}
and with \eqref{eq:bound} the assertion of \eqref{eq:probability} still holds.
Hence, we have shown that there exist $K$ realisations
$\mathcal{C}_i:=\mathcal{C}(\mathcal{Z}_i=\gamma_i)$, $\gamma_i \in \Gamma$, $i \in \{1,\ldots,K\}$ of the
random code, which build a new reduced random code with uniform distribution on these codes with mean
average error probability and mean secrecy criterion fulfilled by \eqref{eq:crit_reduced}. \qed
\end{proof}
Now, if we assume that the channel to the legitimate receiver is non-symmetrisable, which means that
$C(\{W_{s^n}\})>0$, and that there exist a random code $\mathcal{C}^{\textrm{ran}}_n$ that achieves the random
code capacity $C_{S,\textrm{ran}}(\mathfrak{W})>0$, then there exist a sequence of random $(n,J_n)$ codes with
\begin{equation*}
  \lim_{n\to\infty}\max_{s^n \in S^n}  \, \frac{1}{J_n} \sum^{J_n}_{j=1} 
  \sum_{\gamma \in \Gamma} \sum_{x^n \in A^n} 
  E^{\gamma} (x^n| j) \cdot W_{s_n}^{n}((D^{\gamma}_j)^c| x^n) \mu(\gamma) = 0 \enspace, 
\end{equation*}
 \[\liminf_{n\to\infty}\frac{1}{n}\log J_n \to  C_{S,\textrm{ran}}(\mathfrak{W})>0,\] 
and
\begin{equation}\label{eq:rand_code_cap}
\lim_{n\to\infty} \max_{s^n \in S^n} \sum_{\gamma \in \Gamma} I(p_J; V^n_{s^n};\mathcal{C}(\gamma))
\mu(\gamma) =0. 
\end{equation}
Then on account of the random code reduction lemma there exist  a sequence of
random $(n,J_n)$ codes consisting only of $n^3$ deterministic codes (cf. \eqref{eq:assumpt}) chosen from
the former random code, and it holds for any $\epsilon>0$ and sufficiently large $n$ that
\begin{equation}\label{eq:epsilon1}
\max_{s^n \in S^n}  \, \frac{1}{J_n} \sum^{J_n}_{j=1} \frac {1}{n^3} \sum^{n^3}_{i=1} \sum_{x^n \in A^n}  
E^{i} (x^n| j) W_{s_n}^{n}((D^{i}_j)^c| x^n) \leq \epsilon 
\end{equation}
and
\begin{equation}\label{eq:epsiloon2}
 \max_{s^n \in S^n} \frac{1}{n^3}\sum^{n^3}_{i=1} I(p_J; V^n_{s^n};\mathcal{C}_i)\leq \epsilon, 
\end{equation} 
where $\mathcal{C}_i=\{(E^i_j,D^i_j), j\in \mathcal{J}_n\}$, $i=1,\ldots,n^3$, and $E^i$ is the
stochastic encoder of the deterministic wiretap code.  Then the reduced random code consists of the family
of codes $\{\mathcal{C}_i\}_{i \in \{1, \ldots, n^3\}}$ together with the uniform distribution
$\mu'(i)=\frac{1}{n^3}$ for all $i \in \{1,\ldots,n^3\}$.
\subsubsection{Elimination of randomness}(Cf. Theorem $6.11$ in \cite{csis2})\\
Now if there exist a deterministic code and $C(\{W_{s^n}\})> 0$ then there exist a code
\begin{equation}
\{x^{k_n}_i, F_i \subset B^{k_n}: i=1, \ldots n^3 \}
\end{equation}
where $x^{k_n}_i$ is chosen according to an encoding function $f_i:\{1,\ldots,n^3\} \to A^{k_n}$ with $\frac{k_n}{n} \to 0$ as $n \to \infty$ with error probability
\begin{equation}
  \frac{1}{n^3} \sum^{n^3}_{i=1} W^{ k_n}(F^c_i|x^{k_n}_i, s^{k_n}) \leq \epsilon \enspace
\end{equation}
for any $\epsilon > 0$ and sufficiently large $n$ (cf. \eqref{eq:epsilon1}) for all $s^{k_n} \in S^{k_n}$. 
If we now compose a new deterministic code for the AVWC $\mathfrak{W}$ by prefixing the codewords of each
$C_i$
\begin{equation}
  \{f_iE^i_j,F_i \times D^i_j: i=1, \ldots , n^3, 
  j\in [J_n] \}=: \mathcal{C} \enspace,
\end{equation}
the decoder is informed of which encoder $E^i$ is in use for the actual message $j$ if he identifies the prefix
correctly.  Note that for the transmission of the prefix only the reliability is of interest, because it
contains no information about the message $j \in \mathcal{J}_n$ to be sent. Now the new codewords has a
  length of $k_n+n$, transmit a message from $\{1,\ldots, n^3\} \times \mathcal{J}_n$, where the
  channel which is determined by the state sequence $s^{k_n+n} \in S^{k_n+n}$ yields an average error
  probability of  
\begin{equation}\label{eq:43}
\begin{split}
\bar{\lambda}_n(\mathcal{C}, W^{(k_n+n)}_{s^{k_n+n}}) & \leq \frac{1}{n^3 J_n}\sum^{n^3}_{i=1}
\sum_{j \in [J_n]} (\lambda_i+ \lambda_j(i) ) \\
& \leq \frac{1}{n^3 } \sum^{n^3}_{i=1} \lambda_i+ \frac{1}{n^3 } \sum^{n^3}_{i=1}
e_n({s^n},\mathcal{C}_i) \leq 2\epsilon.
\end{split}
\end{equation}
Here, for each $s^{k_n} \in S^{k_n}$ $\lambda_i$ means the error probability for transmitting $i$ from $\{1,\ldots,
n^3\}$ encoded in $x^{k_n}_i$ by $W^{k_n}_{s^{k_n}}$ followed by the transmission of $j$, where the
codeword is chosen according to the stochastic encoder $E^i_j$, over the last $n$ channel realisations
determined by $s^n$ with error probability $\lambda_j(i)$. This construction is possible due to the
memorylessness of the channel. 

Now if we turn to the security part of the transmission problem it is easily seen that
\begin{equation}\label{eq:secrecy}
\begin{split}
p^{\mathcal{C}}_{JZ^{k_n+n}_{s^{k_n+n}}} (j,{z}^{k_n+n}) 
& =\frac{1}{J_n} \frac{1}{n^3}\sum^{n^3}_{i=1}V^{k_n}_{s^{k_n}} (\hat{z}^{k_n} | x^{k_n}_i )
\sum^{}_{x^n} E^i(x^n|j) V^{n}_{s^n} (z^n| x^n)\\
& = \frac{1}{n^3}\sum^{n^3}_{i=1} V^{k_n}_{s^{k_n}} (\hat{z}^{k_n} | x^{k_n}_i ) \cdot
p^{\mathcal{C}_i}_{JZ^n_{s^n}} \enspace,
\end{split}
\end{equation}
where $\hat{z}^{k_n}$ are the first $k_n$ components of ${z}^{k_n+n}$.
With \eqref{eq:secrecy} and the representation of the mutual information by the information divergence we
obtain that
\begin{equation}\label{eq:46}
\begin{split}
& D(p^{\mathcal{C}}_{JZ^{k_n+n}_{s^{k_n+n}}}  || p_J \otimes   p^{\mathcal{C}}_{Z^{k_n+n}_{s^{k_n+n}}})\\
& =D\Big(\frac{1}{n^3}\sum^{n^3}_{i=1}  V^{k_n}_{s^{k_n}} (\hat{z}^{k_n} | x^{k_n}_i ) p^{\mathcal{C}_i}_{JZ^n_{s^n}} \Big\Vert
\frac{1}{n^3}\sum^{n^3}_{i=1}  V^{ k_n}_{s^{k_n}} (\hat{z}^{k_n} | x^{k_n}_i ) p_J \otimes p^{\mathcal{C}_i}_{Z^n_{s^n}}\Big) \\
& \leq \frac{1}{n^3}\sum^{n^3}_{i=1} D\big(  V^{k_n}_{s^{k_n}}  (\hat{z}^{k_n} | x^{k_n}_i ) p^{\mathcal{C}_i}_{JZ^n_{s^n}} \big\Vert
V^{k_n}_{s^{k_n}} (\hat{z}^{k_n} | x^{k_n}_i ) p_J \otimes p^{\mathcal{C}_i}_{Z^n_{s^n}}\big) \\
&=\frac{1}{n^3}\sum^{n^3}_{i=1} D\big(  p^{\mathcal{C}_i}_{JZ^n_{s^n}} \big\Vert
p_J \otimes p^{\mathcal{C}_i}_{Z^n_{s^n}}\big)
=\frac{1}{n^3}\sum^{n^3}_{i=1} I(p_J ,  V^{n}_{s^n} ;  \mathcal{C}_i) \leq \epsilon
\end{split}
\end{equation}
for all $s^n \in S^n$ and $n \in \mathbb{N}$ sufficiently large, where the first inequality follows because for two probability distributions $p,q$ the
relative entropy $D(p\Vert q)$ is a convex function in the pair $(p,q)$ and the last inequality follows by
the random code reduction lemma.

Because $\frac{k_n}{n} \to 0$ as $n \to \infty$
\begin{equation}
\lim_{n \to \infty} \frac{1}{k_n+n} \log(n^3J_n)  =\lim_{n \to \infty} (\frac{1}{n} \log J_n +\frac{1}{n} \log(n^3))
 = \lim_{n \to \infty} \frac{1}{n} \log J_n \enspace,
\end{equation}
$\mathcal{C}_n$ is a deterministic $(n,J_n)$ code which achieves the same rates as the random code
$\mathcal{C}^{\textrm{ran}}_n$ and so the random code capacity $C_{S,\textrm{ran}}$ as given in
\eqref{eq:rand_code_cap}, provided that the channel to the
legitimate receiver is non-symmetrisable. 

Thus, with $\{1, \ldots,  J_n\}$ as the message set, $\mathcal{C}_n$ is a deterministic $(n+o(n),n^3 \cdot
J_n)$ code with average error probability bounded for all $s^{k_n+n} \in S^{k_n+n}$ as in \eqref{eq:43} and which
fulfills the strong secrecy criterion as in \eqref{eq:46}, and which achieves the random code secrecy
capacity $C_{S,\textrm{ran}}$ of the arbitrarily varying wiretap channels AVWC $\mathcal{W}$ which
implies that $C_S=C_{S,\textrm{ran}}$. This concludes the proof.

Note that in the case in which the channel to the legitimate receiver is non-symmetrisable and we know
that the deterministic code secrecy capacity $C_S(\mathfrak{W})$ equals zero we can conclude that the random
code secrecy capacity $C_{S,\textrm{ran}}(\mathfrak{W})$ equals zero. As a consequence of the theorem we
can state the following assertion.
\begin{corollary}\label{lower_bound}
The deterministic code secrecy capacity of the arbitrarily varying wiretap channel $\mathfrak{W}$,
provided that there exists a best channel to the eavesdropper and under the assumption that the channel
to the legitimate receiver is non-symmetrisable, is lower bounded by
\begin{equation*}
C_S(\mathfrak{W}) \geq \max_{p \in \mathcal{P}(A)} ( \min_{q \in \mathcal{P}(S)}
I(p,W_{q})-\max_{q \in \mathcal{P}(S)}I(p,V_{q})) \enspace.
\end{equation*}
\end{corollary}
\begin{proof}
Combine the assertions of Proposition \ref{random_code} and Theorem \ref{achievability}. \qed
\end{proof}

\subsection{Upper bound on the capacity of the AVWC $\mathfrak{W}$ and a multi-letter coding theorem}\label{upper_bound} 

In this section we give an upper bound on the secrecy capacity of the AVWC $\mathfrak{W}$ which
corresponds to the bound for the compound wiretap channel built by the same family of channels. In
addition we give the proof of the multi-letter converse of the AVWC $\mathfrak{W}$.
\begin{theorem}\label{upper_b}
The secrecy capacity of the arbitrarily varying wiretap channel AVWC $\mathfrak{W}$ is upper bounded,
\begin{equation}
C_S(\mathfrak{W}) \leq \min_{q \in \mathcal{P}(S)} \max_{U \to X \to (YZ)_q} (I(U,Y_q)- I(U,Z_q)) \enspace.
\end{equation}
\end{theorem}
\begin{proof}
By Lemma \ref{cap_closure} the capacity of the AVWC $\mathfrak{W}$ equals the capacity of the AVWC
$\overline{\mathfrak{W}}$. Obviously, the set $\overline{\mathcal{W}} = \{(W^{\otimes n}_q, V^{\otimes n}_q): q
\in \mathcal{P}(S)\}$ which describes a compound wiretap channel is a subset of $\overline{\mathfrak{W}}^n
= \{(W^{n}_{\tilde{q}}, V^{n}_{\tilde{q}}): \tilde{q} \in \mathcal{P}(S^n), \tilde{q}=\prod^n_{i=1} q_i
\}$. Now, because we can upper bound the secrecy capacity of the AVWC
$\overline{\mathfrak{W}}$ by the secrecy capacity of the worst wiretap channel in the family
$\overline{\mathfrak{W}}^n$, together with the foregoing we can upper bound it by the capacity of the
worst channel of the compound channel $\overline{\mathcal{W}}$. Hence,
\begin{equation*}
\begin{split}
C_S(\mathfrak{W}) = C_S(\overline{\mathfrak{W}}) 
& \leq \inf_{\tilde{q}} C_S((W^{n}_{\tilde{q}}, V^{n}_{\tilde{q}}))\\ 
& \leq \inf_{{q}} C_S((W^{n}_{{q}}, V^{n}_{{q}}))=\inf_q C_S(W_q,V_q) \enspace,
\end{split}
\end{equation*}
The minimum is attained because of the continuity of $C_S(W_q,V_q)$ on the compact set
$\overline{\mathfrak{W}}$. \qed
\end{proof}
\begin{remark}
Consider the special case of an AVWC $\mathfrak{W}=\{(W_{s^n},V_{r^n}):s^n \in S_1^n,\ r^n \in
S_2^n\}$, where both the state of the main channel $s \in S_1$ and the state of the eavesdropper's channel
$r \in S_2$ in every time step can be chosen independently. In addition let us assume that there exist a
channel $W_{q^*_1} \in \{W_{q_1}: q_1 \in \mathcal{P}(S_1)\}$, which is a degraded version of all other
channels from $\{W_{q_1}: q_1 \in \mathcal{P}(S_1)\}$, and a best channel to the eavesdropper $V_{q^*_2}$
from the set $\{V_{q_2}: q_2 \in \mathcal{P}(S_2)\}$ (cf. Definition \ref{best}). Then in
accordance with Section $3.5$ of \cite{bjela2} the lower bound on the secrecy capacity given in Corollary
\ref{lower_bound} matches the upper bound from Theorem \ref{upper_b}. Thus we can conclude that under the
assumption, that the channel to legitimate receiver is non-symmetrisable, the capacity of the AVWC
$\mathfrak{W}$ is given by 
\begin{equation*}
C_S(\mathfrak{W})=\max_{p \in \mathcal{P}(A)} (I(p,W_{q^*_1})-I(p,V_{q^*_2})) \enspace.
\end{equation*}
\end{remark}
Now in addition to Theorem \ref{upper_b} we give a multi-letter formula of the upper bound of the secrecy
rates. Therefore we need the following lemma used in analogy to Lemma $3.7$ in \cite{bjela2}.
\begin{lemma}
For the arbitrarily varying wiretap channel AVWC $\mathfrak{W}^n$ the limit
\begin{equation*}
\lim_{n \to \infty} \frac{1}{n} \max_{U \to X^n \to (Y^nZ^n)_{\tilde{q}}} (\inf_{\tilde{q} \in
  \mathcal{P}(S^n)} I(U,Y^n_{\tilde{q}}) - \sup_{\tilde{q} \in \mathcal{P}(S^n)} I(U,Z^n_{\tilde{q}}))
\end{equation*}
exists.
\end{lemma}
The proof is carried out in analogy to Lemma $3.7$ in \cite{bjela2} and therefore omitted.
\begin{theorem}
The secrecy capacity of the arbitrarily varying wiretap channel AVWC $\mathfrak{W}$ is upper bounded by
\begin{equation}
C_S(\mathfrak{W})  \leq 
 \lim_{n \to \infty} \frac{1}{n} \max_{U \to X^n \to (Y^nZ^n)_{\tilde{q}}} (\inf_{\tilde{q} \in
  \mathcal{P}(S^n)} I(U,Y^n_{\tilde{q}}) - \sup_{\tilde{q} \in \mathcal{P}(S^n)} I(U,Z^n_{\tilde{q}})) \enspace, 
\end{equation}
where $\tilde{q}=\prod^n_{i=1} q_i$, $q_i \in \mathcal{P}(S)$ and $Y^n_{\tilde{q}}, Z^n_{\tilde{q}}$ are
the outputs of the channels $W^n_{\tilde{q}}$ and $V^n_{\tilde{q}}$ respective.
\end{theorem}
\begin{proof}
Let $(\mathcal{C}_n)_{n \in \mathbb{N}}$ be any sequence of $(n,J_n$) codes such that with
\begin{equation*}
\sup_{s^n \in S^n}  \, \frac{1}{J_n} \sum^{J_n}_{j=1} \sum_{x^n \in A^n} 
E(x^n| j) W_{s_n}^{n}(D_j^c| x^n)=:\varepsilon_{1,n} \ \mathrm{and}, \
\sup_{s^n \in S^n} I(J,Z^n_{s^n})=: \varepsilon_{2,n}
\end{equation*}
it holds that $\lim_{n \to \infty} \varepsilon_{1,n}0=$ and $\lim_{n \to \infty} \varepsilon_{2,n}$,
where $J$ denotes the random variable which is uniformly distributed on the message set $\mathcal{J}_n$.
Because of Lemma \ref{cap_closure} we obtain that for the same sequences of $(n,J_n)$ codes
\begin{equation}\label{eq:error_bound}
\lim_{n \to \infty} \sup_{\tilde{q} \in \mathcal{P}(S^n)}  \, \frac{1}{J_n} \sum^{J_n}_{j=1} \sum_{x^n \in A^n} 
E(x^n| j)  W_{\tilde{q}}^{n}(D_j^c| x^n)
=\lim_{n \to \infty} \varepsilon_{1,n}=0
\end{equation}
and
\begin{equation}\label{eq:secr_clo}
\lim_{n\to \infty} \sup_{\tilde{q} \in \mathcal{P}(S^n)} I(J,Z^n_{\tilde{q}})=\lim_{n \to \infty}
\varepsilon_{2,n}=0 \enspace.
\end{equation}
Now let us denote another random variable by $\hat{J}$ with values in $\mathcal{J}_n$ determined by the
Markov chain $J \to X^n \to Y^n_{\tilde{q}} \to \hat{J}$, where the first transition is governed by $E$,
the second by $W^n_{\tilde{q}}$, and the last by the decoding rule. Now the proof is analogue to the
proof of Proposition $3.8$ in \cite{bjela2}. For any $\tilde{q} \in \mathcal{P}(S^n)$ we have from data
processing and Fano's inequality
\begin{equation*}
(1-\varepsilon_{1,n}) \log J_n \leq I(J,Y^n_{\tilde{q}})+1.
\end{equation*}
We then use the validity of the secrecy criterion \eqref{eq:secr_clo} to derive
\begin{equation*}
(1-\varepsilon_{1,n}) \log J_n \leq I(J,Y^n_{\tilde{q}}) - \sup_{\tilde{q}} I(J,Z^n_{\tilde{q}}) +
  \varepsilon_{2,n}+1
\end{equation*}
for any $\tilde{q} \in \mathcal{P}(S^n)$. Since the LHS does not depend on $\tilde{q}$ we end in
\begin{equation*}
  (1 -  \varepsilon_{1,n}) \log J_n 
  \leq \max_{U \to X^n \to Y^n_{\tilde{q}}Z^n_{\tilde{q}}} (\inf_{\tilde{q}}
  I(U,Y^n_{\tilde{q}}) - \sup_{\tilde{q}} I(U,Z^n_{\tilde{q}})) +
  \varepsilon_{2,n}+1 \enspace.
\end{equation*}
Dividing  by $n \in \mathbb{N}$ and taking $\limsup$ concludes the proof. \qed
\end{proof}
Now if we consider the set $\overline{\mathcal{W}} = \{(W^{\otimes n}_q, V^{\otimes n}_q): q
\in \mathcal{P}(S)\}$ as a subset of  $\overline{\mathfrak{W}}^n
= \{(W^{n}_{\tilde{q}}, V^{n}_{\tilde{q}}): \tilde{q} \in \mathcal{P}(S^n), \tilde{q}=\prod^n_{i=1} q_i
\}$ and the same sequence $(\mathcal{C}_n)_{n \in \mathbb{N}}$ of $(n,J_n$) codes for the AVWC
$\mathfrak{W}$ for which \eqref{eq:error_bound} and \eqref{eq:secr_clo} holds, we can conclude that
\begin{equation}\label{}
\lim_{n \to \infty} \sup_{q \in \mathcal{P}(S)}  \, \frac{1}{J_n} \sum^{J_n}_{j=1} \sum_{x^n \in A^n} 
E(x^n| j)  W_{q}^{\otimes n}(D_j^c| x^n)
\leq \lim_{n \to \infty} \varepsilon_{1,n}
\end{equation}
and
\begin{equation}\label{}
\lim_{n\to \infty} \sup_{{q} \in \mathcal{P}(S)} I(J,Z^n_{{q}}) \leq \lim_{n \to \infty} \varepsilon_{2,n} \enspace,
\end{equation}
with $\varepsilon_{1,n}$ and $\varepsilon_{2,n}$ as above. Then we can conclude with the same
argumentation as in the previous proof,
\begin{corollary}
The secrecy capacity of the arbitrarily varying wiretap channel AVWC $\mathfrak{W}$ is upper bounded by
\begin{equation*}
 C_S(\mathfrak{W})  \leq  \lim_{n \to \infty} \frac{1}{n} \max_{U \to X^n \to (Y^nZ^n)_{{q}}} (\inf_{{q} \in
  \mathcal{P}(S)} I(U,Y^n_{{q}}) - \sup_{{q} \in \mathcal{P}(S)} I(U,Z^n_{{q}})) \enspace, 
\end{equation*}
where $q \in \mathcal{P}(S)$ and $Y^n_{{q}}, Z^n_{{q}}$ are
the outputs of the channels $W^{\otimes n}_{q}$ and $V^{\otimes n}_{{q}}$ respective.
\end{corollary}
Now, using standard arguments concerning the use of the channels defined by $P_{Y_q|U}=W_q \cdot P_{X|U}$
and $P_{Z_q|U}=V_q \cdot P_{X|U}$ instead of $W_q$ and $V_q$ and applying the assertion of Corollary
\ref{lower_bound} to the $n$-fold product of channels $W_q$ and $V_q$, we are able to give the coding theorem
for the multi-letter case of the AVWC with a best channel to the eavesdropper.
\begin{theorem}
Provided that there exist a best channel to the eavesdropper, the multi-letter expression for the secrecy capacity
$C_S(\mathfrak{W})$ of the AVWC $\mathfrak{W}$ is given by
\begin{equation*}
 C_S(\mathfrak{W}) =  
 \lim_{n \to \infty} \frac{1}{n} \max_{U \to X^n \to (Y^nZ^n)_{{q}}} (\inf_{{q} \in
  \mathcal{P}(S)} I(U,Y^n_{{q}}) - \sup_{{q} \in \mathcal{P}(S)} I(U,Z^n_{{q}})) \enspace, 
\end{equation*}
if the channel to the legitimate receiver is non-symmetrisable, and is zero otherwise.
\end{theorem}

\section*{Acknowledgment}
Support by the Deutsche Forschungsgemeinschaft (DFG) via projects 
BO 1734/16-1, BO 1734/20-1, and by the Bundesministerium f\"ur Bildung und Forschung (BMBF) via grant
01BQ1050 is gratefully acknowledged.

\bibliographystyle{splncs03}
\bibliography{references}

\end{document}